\documentclass[11pt,a4paper]{article}%
\usepackage{amsfonts}
\usepackage{amsmath}
\usepackage{amssymb}
\usepackage{graphicx}
\usepackage{hyperref}
\usepackage{geometry}
\usepackage{booktabs} 
\usepackage{array}
\usepackage{cancel}%
\providecommand{\U}[1]{\protect\rule{.1in}{.1in}}
\newtheorem{theorem}{Theorem}

\newtheorem{proposition}{Proposition}
\newtheorem{remark}{Remark}

\newenvironment{proof}[1][Proof]{\textbf{#1.} }{\  \rule{0.5em}{0.5em}}
\makeatletter
\def \@removefromreset#1#2{\let \@tempb \@elt
\def \@tempa#1{@&#1}\expandafter \let \csname @*#1*\endcsname \@tempa
\def \@elt##1{\expandafter \ifx \csname @*##1*\endcsname \@tempa \else
\noexpand \@elt{##1}\fi}     \expandafter \edef \csname cl@#2\endcsname{\csname cl@#2\endcsname}     \let \@elt \@tempb
\expandafter \let \csname @*#1*\endcsname \@undefined}

\@removefromreset{equation}{section}

\@removefromreset{theorem}{section}
\makeatother
\usepackage{listings,xcolor}
\begin{document}

\title{From Majorana to Bloch: two ways to generalize the Bloch sphere representation}
\author{Louis Hanotel\\
Tikhonov Moscow Institute of Electronics and Mathematics, HSE University, \\
Moscow 123458, Russia}

\maketitle

\begin{abstract}
While the Bloch sphere provides a canonical geometric representation for qubits, its generalization to pure qudit states ($d>2$) is not unique. This representation can be extended following two different approaches: increasing the number of points, as in the Majorana representation, or using a single point but in a space of larger dimension, as in the generalized Bloch representation. In this work, we compare these representations of pure qudit states by finding explicit analytical relations between them. In particular, a general formula of the Bloch vector in terms of the Majorana points of a given state is provided.  We also analyze some suitable extensions of these representations in the case of pure permutation-symmetric two- and three-qubit states, deriving the corresponding relation between them. These analytical results allow us to study, from a geometric point of view, quantum correlations of special classes of two-qubit states.
\end{abstract}

%\tableofcontents
%\pagebreak

\section{Introduction}
A pure state of a $d$-level quantum system is represented by a vector in a Hilbert space $\mathcal{H}_d \cong \mathbb{C}^d$. Any state, pure or mixed, can be described by a density operator acting on $\mathcal{H}_d$. Density operators are Hermitian, positive semidefinite and unit-trace. Let $\mathcal{B}(\mathcal{H}_d)$ denote the space of bounded linear operators acting on $\mathcal{H}_d$. The space of density operators of $d$-level quantum systems is a convex set defined as
\begin{equation}
    \mathcal{L}_{+,1}(\mathcal{H}_d)=
    \left\lbrace
    \rho\in\mathcal{B}(\mathcal{H}_d)\ :\rho=\rho^{\dagger},\ \mathrm{tr}\ \rho=1,\ \rho\geq 0\ \right\rbrace.
\end{equation}

For $d=2$, the most commonly used representation of quantum states in $\mathcal{L}_{+,1}(\mathcal{H}_2)$---both pure and mixed---is the Bloch representation. In it, a pure qubit state is represented by a point on the surface of the unit sphere, known in quantum information theory as the \emph{Bloch sphere}. This point can be defined operationally, for example, in the case of spin systems, as the direction in space defined by the spin expectation value in a given state. On the other hand, a mixed state is represented by a point inside the unit sphere, i.e., a point in the \emph{Bloch ball}. 

Nevertheless, even for pure qudit states of $d>2$, there is no unique extension of the Bloch  representation. There are at least two ways to generalize the Bloch sphere to represent pure qudit states in geometric terms: we can describe the state by a single point on a space of larger dimension, as done in the generalized Bloch representation \cite{Kim:03, Lou.Kul:21}, or we can allow for a description on the unit sphere in three-dimensional space while increasing the number of points --- as done in the Majorana representation \cite{Majo:32}. For mixed states, the Majorana-like representations require different sets of points and, in general, spheres of different radii \cite{Ser.Bra:20}. In this work, we analyze how those two approaches are related, motivated by their potential applications to analyze quantum correlations, as exemplified below; but also because these relations could be used to reduce the number of parameters in some optimization tasks of interest in quantum information.

The generalized Bloch representation, hereafter Bloch representation, associates every qudit state, pure or mixed, with a point in $\mathbb{R}^{d^2-1}$, referred to as its \emph{Bloch vector}, whose components can be defined operationally, making it suitable for experimental contexts \cite{Kim:03}. In contrast, the Majorana representation associates every pure qudit state with a set of $d-1$ points on the unit sphere $\mathcal{S}^2\subset\mathbb{R}^{3}$ that are usually referred to as \emph{stars}, and the whole set as a \emph{(Majorana) constellation}. The Majorana representation has proven particularly valuable for analyzing the rotational properties of quantum states as it faithfully encodes the action of rotations on quantum states. It is worth to mention that there are other alternative representations of quantum states in terms of points of a suitable (set of) sphere(s), e. g., the representation by zeroes of the Husimi function \cite{Leb.Vor:90}, the one of Ramachandran and Ravishankar \cite{Ram.Rav:86}---based on polarization tensors---and the related generalized Majorana representation for arbitrary Hermitian operators \cite{Ser.Bra:20}, then including all mixed states but requiring a multiconstellation, i. e., a set of constellations rather than a single one, in spheres of different radii. For certain classes of multipartite states with special symmetries, generalizations of the Majorana representation have also been proposed in terms of multiconstellations, as in \cite{Chr.Guz.etal:21,Chr.Han.etal:21}. Other related representations with different ranges of applicability are found in \cite{Kur:11, Kur.etal:16, Elt.etal:21, Sha.Goz.Saz:24} and references therein.

In this work, we compare these representations in the case of pure qudit states, in particular, by finding a general expression of the Bloch vector in terms of the Majorana constellation. We also derive specific results for the case of pure permutation symmetric states of two and three qubits, since both representations admit a natural extension to those classes of quantum states without introducing the multiconstellation formalism used in case of general states. The relation between representations proves to be useful for analyzing quantum correlations and to reduce the number of parameters in optimization tasks involving rotationally invariant functions.

The article is organized as follows. In Section \ref{Section2}, Bloch and Majorana representations are briefly presented and, in Section \ref{Section3}, the relation between them for pure qudit states is discussed, in particular, Theorem \ref{Theorem1} in this section provides an exact analytical expression for the Bloch vector in terms of the Majorana constellation of the state. Analogous relations are also found in the case of symmetric multiqubit states; specifically, in Section \ref{Section4}, we analyze how to express the corresponding
%generalized Bloch representation parameters?
Bloch representation parameters in terms of the Majorana points in the case of pure permutation-symmetric two- and three-qubit states, as formalized in Theorems \ref{Theorem2} and \ref{Theorem3}. As an application of these results, in Section \ref{Section5}, we analyze the quantum correlations, quantified via the geometric quantum discord, of permutation symmetric two-qubit states in terms of Majorana stars. With these results we also find  a connection between a spin anticoherence measure and geometric quantum discord for permutation symmetric pure  two-qubit states, as we formulate in Theorem \ref{Theorem4}. In Section \ref{Section6}, conclusions and final remarks are presented.

\section{Preliminaries: Bloch and Majorana representations of qudit states}
\label{Section2}
In this section we briefly recall the Bloch representation for arbitrary qudit states and the Majorana representation of pure qudit states.

\subsection{Bloch representation}

In the Bloch representation any qudit density matrix is expressed by a vector in  $\mathbb{R}^{d^2-1}$, whose components are related to the generalized Gell-Mann operators, which constitute the standard basis of the Lie algebra $\mathfrak{su}(d)$ in its defining representation, consisting of $d^2-1$ matrices of size $d\times d$.

Consider the $(d^2-1)$-dimensional tuple $\Lambda_d=\lbrace \Lambda_d^{(i)}\rbrace_{i=1}^{d^2-1}$ of the aforementioned $\mathrm{SU}(d)$ generators, ordered as follows: first the symmetric generators, then antisymmetric ones and finally the diagonal ones. The tuple $\Lambda_d$ can then be expressed as
\begin{equation}
\label{gen_sasd}
\Lambda_d=\left(\Lambda_{12}^{(\mathrm{s})}, \ldots, \Lambda_{1 d}^{(\mathrm{s})}, \ldots, \Lambda_{d-1\phantom{,} d}^{(\mathrm{s})}, \Lambda_{12}^{(\mathrm{a s})}, \ldots, \Lambda_{1 d}^{(\mathrm{a s})}, \ldots, \Lambda_{d-1\phantom{,} d}^{(\mathrm{a s})}, \Lambda_1^{(\mathrm{diag})}, \ldots, \Lambda_{d-1}^{(\mathrm{diag})}\right),
\end{equation}
 where the generators in the computational basis take the form \cite{Lou:20}:
\begin{align}
\label{lambda_sym}
\Lambda_{m k}^{(\mathrm{s})} & =|m\rangle\langle k|+|k\rangle\langle m|, \quad \Lambda_{m k}^{(\mathrm{s})}=\Lambda_{k m}^{(\mathrm{s})}, \quad 1 \leq m<k \leq d, \\
\label{lambda_asym}
\Lambda_{m k}^{(\mathrm{as})} & =-i|m\rangle\langle k|+i|k\rangle\langle m|, \quad \Lambda_{m k}^{(\mathrm{as})}=-\Lambda_{k m}^{(\mathrm{as})}, \quad 1 \leq m<k \leq d \\
\label{lambda_diag}
\Lambda_l^{(\mathrm{diag})} & =\sqrt{\frac{2}{l(l+1)}}\left(\sum_{j=1}^l|j\rangle\langle j|-l|l+1\rangle\langle l+1|\right), \quad 1 \leq l \leq d-1\ ,
\end{align}
and satisfy, for any values of $k,m=1,2, \cdots,d^2-1$, the relation

\begin{equation}
\label{prodSUd}
    \Lambda^{(k)}_d \Lambda^{(m)}_d
    =
    \frac{2}{d}\delta_{km}\mathbb{I}_{\mathbb{C}^d}
    + 
    \sum_{l=1}^{d^2-1}
    \left(
    g_{kml}+if_{kml}
    \right)
    \Lambda^{(l)}_d \ ,
\end{equation}
where $g_{kml}:=\frac{1}{4}\mathrm{tr}\left(\lbrace \Lambda^{(k)}_d, \Lambda^{(m)}_d\rbrace\Lambda^{(l)}_d\right)$ denote the symmetric coefficients and $f_{kml}:=\frac{1}{4i}\mathrm{tr}\left([ \Lambda^{(k)}_d, \Lambda^{(m)}_d]\Lambda^{(l)}_d\right)$ the (antisymmetric) structure constants of the Lie algebra $\mathfrak{su}(d)$, $\delta_{km}$ denotes the Kronecker delta,  $\mathbb{I}_{\mathbb{C}^d}$ the identity matrix in $\mathbb{C}^d$, and $\lbrace \cdot,\cdot\rbrace$ and $[\cdot,\cdot]$ denote the anticommutator and commutator of two operators, respectively.

The $d^2-1$ operators in \eqref{lambda_sym}-\eqref{lambda_diag} constitute a basis for the space of traceless Hermitian operators acting on $\mathbb{C}^d$, generalizing the Pauli and Gell-Mann operators for the cases $d=2,3$, respectively. Therefore, any qudit state $\rho\in \mathcal{L}_{+,1}(\mathbb{C}^d)$ can be written in the form
\begin{equation}
\label{rho_gmrep}
    \rho=\frac{1}{d}\mathbb{I}_{\mathbb{C}^d}+\sqrt{\frac{d-1}{2d}}r_\rho\cdot \Lambda_d\ ,
\end{equation}
where the components $r_{\rho}^{(i)}$ of the vector $r_\rho\in\mathbb{R}^{d^2-1}$ are defined as
\begin{equation}
\label{bloch_components}
    r_\rho^{(i)}=\sqrt{\frac{d}{2(d-1)}}\mathrm{tr}(\Lambda^{(i)}_d \rho),\ \  i=1,2,\dots,d^2-1\ .
\end{equation}

The vector $r_{\rho}\in\mathbb{R}^{d^2-1}$, known as \emph{Bloch vector}, in the chosen basis, can be also expressed as
\begin{equation}
\label{r_stb1}
r_{\rho}=
\left(r_{12}^{(\mathrm{s})}, \ldots, r_{1 d}^{(\mathrm{s})}, \ldots, r_{d-1\phantom{,} d}^{(\mathrm{s})},r_{12}^{(\mathrm{as})}, \ldots, r_{1 d}^{(\mathrm{as})}, \ldots, r_{d-1\phantom{,} d}^{(\mathrm{as})}, r_1^{(\mathrm{diag})}, \ldots, r_{d-1}^{(\mathrm{diag})}\right),
\end{equation}
where
\begin{equation}
\label{r_stb2}
r_{mk}^{(\mathrm{s})}=\sqrt{\frac{d}{2(d-1)}}\mathrm{tr}(\rho \Lambda_{mk}^{(\mathrm{s})}), \ r_{mk}^{(\mathrm{as})}=\sqrt{\frac{d}{2(d-1)}}\mathrm{tr}(\rho \Lambda_{mk}^{(\mathrm{as})}),\  r_{l}^{(\mathrm{diag})}=\sqrt{\frac{d}{2(d-1)}}\mathrm{tr}(\rho \Lambda_{l}^{(\mathrm{diag})}).
\end{equation}
%\begin{equation}
%\label{r_stb2}
%r_{mk}^{(\mathrm{s})}=\sqrt{\frac{d}{2(d-1)}}\mathrm{tr}(\rho \Lambda_{mk}^{(\mathrm{s})}),  r_{mk}^{(\mathrm{as})}=\sqrt{\frac{d}{2(d-1)}}\mathrm{tr}(\rho \Lambda_{mk}^{(\mathrm{as})}), r_{l}^{(\mathrm{diag})}=\sqrt{\frac{d}{2(d-1)}}\mathrm{tr}(\rho \Lambda_{l}^{(\mathrm{diag})}).
%\end{equation}
So, in this notation, we have $r_{\rho}^{(1)}=r_{12}^{(\mathrm{s})},\dots,r_{\rho}^{(d^2-1)}=r_{d-1}^{(\mathrm{diag})}$.

In this representation, the maximally mixed state has a vanishing Bloch vector; the Euclidean distance in $\mathbb{R}^{d^2-1}$ between the points representing two states is also equal (or proportional, depending on the chosen normalization) to the Hilbert-Schmidt distance between the corresponding density operators in $\mathcal{B}(\mathcal{H})$ \cite{Ben.Zyc:2017};  the purity condition  $\rho^2=\rho$ is equivalent, by \eqref{prodSUd}, to the following conditions:
\begin{proposition}
\label{Proposition1}
    A qudit density matrix $\rho \in \mathcal{L}_{+,1} (\mathbb{C}^{d})$  represents a pure state if and only if its Bloch vector $r_\rho \in \mathbb{R}^{d^2-1}    $ satisfies the following conditions \cite{Ben.Zyc:2017, Arv.Mal.Muk:97}:
\begin{itemize}
    \item[(i)] $\|r_\rho\|_{\mathbb{R}^{d^2-1}}=1,$
    \item[(ii)] $ \sum_{jk} g_{ijk} r_\rho^{(j)} r_\rho^{(k)}=(d-2)\sqrt{\frac{2}{d(d-1)}}r_\rho^{(i)}$, where $g_{ijk}$ are the symmetric constants defined above.
\end{itemize}
\end{proposition}

The above choice of orthogonal basis for the $\mathfrak{su}(d)$ generators is arbitrary in the sense that we could choose another orthogonal basis with respect to the Hilbert-Schmidt inner product in $\mathcal{B}(\mathcal{H})$, defined by $\langle A,B\rangle=\frac{1}{2}\mathrm{tr}(A^{\dagger}B)$. The general Bloch vector formalism ---independent of the particular choice of an orthogonal basis of $\mathfrak{su}(d)$ generators--- is known as the \emph{generalized Bloch representation} and is thoroughly analyzed in \cite{Lou.Kul:21}; some particular representations derived from particular choices of bases are studied in \cite{Ber.Kra:08}. However, in this article we restrict ourselves to the standard Bloch representation  \eqref{rho_gmrep} defined using the operators \eqref{lambda_sym}-\eqref{lambda_diag}.

The Bloch representation can be extended to arbitrary states of any dimension and number of parties, but the geometry of the set of possible Bloch vectors is much more complex than that of the Bloch ball. In \cite{Goy.etal:16} and \cite{Jak.Sie:01}, the geometries of the simplest extensions in dimension and number of parties, respectively, are analyzed. Let us exemplify what this representation looks like in the case of two-qudit states: a general two-qudit density matrix  $\rho \in\mathcal{L}_{+,1} (\mathbb{C}^{d} \otimes \mathbb{C}^d)$ can be written \cite{Lou.Kul:21} in the form 
\begin{equation}
 \label{2qudits}
    \rho=
    \frac{\mathbb{I}_{\mathbb{C}^d} \otimes \mathbb{I}_{\mathbb{C}^d}}{d^2}
    +
    \sqrt{\frac{d-1}{2d^3}}
    (r_{\rho} \cdot \Lambda_d)\otimes \mathbb{I}_{\mathbb{C}^d}
    +
    \sqrt{\frac{d-1}{2d^3}}
    \mathbb{I}_{\mathbb{C}^d}\otimes (s_{\rho} \cdot \Lambda_d)
    +
    \frac{1}{4}
    \sum_{i,j=1}^{d^2-1}T_{\rho}^{(ij)}\Lambda^{(i)}_d\otimes\Lambda^{(j)}_d\ ,
\end{equation}

\noindent where  $r_{\rho}^{(i)}:=\sqrt{\frac{d}{2(d-1)}}\mathrm{tr}\left(\rho \ \Lambda^{(i)}_d\otimes\mathbb{I}_{\mathbb{C}^d}\right)$ and $s_{\rho}^{(i)}:=\sqrt{\frac{d}{2(d-1)}}\mathrm{tr}\left(\rho\  \mathbb{I}_{\mathbb{C}^d} \otimes \Lambda^{(i)}_d\right)$ are the \emph{Bloch vectors} of the reduced state density matrices tracing over the corresponding spaces, and $T^{(ij)}_{ \rho}:=\mathrm{tr}\left(\rho \ \Lambda^{(i)}_d\otimes \Lambda^{(j)}_d\right)$ is known as \emph{correlation matrix}.

\subsection{Majorana representation}
The Majorana representation \cite{Majo:32} allows for a representation of a pure qudit state by a set of $d-1$ points on the unit sphere $\mathcal{S}^2\subset\mathbb{R}^3$, generalizing the Bloch sphere representation of qubits to which it reduces in the case of $d=2$. 

The set of points associated with a pure qudit state written in the standard basis, $|\psi\rangle=\sum_{m=1}^{d}\psi^{(m)} |m\rangle\in \mathbb{C}^{d}$, in the Majorana representation is found by computing the roots of the so-called \emph{Majorana polynomial} of the state, on the auxiliary complex variable $\zeta$ of degree $d-1$,
%%%%%% urge revisar la ecuación 12
\begin{equation}
\label{majo}
    P_{|\psi\rangle}(\zeta)=\sum_{m=1}^{d} (-1)^{m} \sqrt{\binom{d-1}{m-1}}\psi^{(m)} \zeta^{d-m}\ ,
\end{equation}
and then finding its $d-1$ roots (with possible degeneration) in the complex plane and projecting them stereographically
%\footnote{The (inverse) stereographic projection from the south pole maps a point $x+iy$ in the complex plane to the point with coordinates $(2x,2y,1-x^2-y^2)/(1+x^2+y^2)$ on the unit sphere.} 
from the south pole to the unit sphere. With this procedure, we can associate with every qudit state, up to a global phase, a set of $d-1$ points on the sphere. Each of those points is commonly referred to as a \emph{star}, and the set of them, a \emph{(Majorana) constellation}. 

Conversely, given a constellation, it is possible to find the corresponding state by reversing the procedure, since the representation is unique modulo a global phase: first, by means of the stereographic projection, map the constellation to a set of points $\zeta_i\in\mathbb{C}$ in the complex plane and then find the coefficient of the state by comparing term by term the Majorana polynomial with the polynomial $\prod_i(\zeta-\zeta_i)$ and solving the corresponding linear equations.

The Majorana representation can be naturally extended to describe permutation-symmetric multiqubit states by exploiting their correspondence with qudit states. All basis states in the standard computational basis of a qudit can be put in a one-to-one correspondence with the (permutation-symmetric) Dicke states \cite{Dev.Sud.Rah:12}. A pure qudit state with a given constellation with stars at $\hat{n}_1,\hat{n}_2,\dots,\hat{n}_{d-1}\in\mathcal{S}^2\subset\mathbb{R}^3$ can be thought as a pure permutation-symmetric state given by the symmetrized tensor product of $d-1$ single qubit states $|\psi_{\hat{n}_i}\rangle$, each represented by a single star at $\hat{n}_i$, for $i=1,2,\dots, d-1$. In the context of spin states, this correspondence is exactly the isomorphism between the Hilbert space of a spin-$s$ particle and that of $2s$ spin-$1/2$ particles. It was precisely in this context that Majorana originally introduced \cite{Majo:32} this geometric representation of quantum states.

Both representations have specific properties that make them useful for different tasks. In the case of the Majorana representation, its main property, besides providing a visualization of quantum states in three-dimensional space, despite the fact that they live in higher-dimensional Hilbert spaces, is that it faithfully represents the action of the rotation group on the projective\footnote{It is a projective representation since any two states differing by a global phase are represented by the same Majorana constellation. So the relation between constellations is one-to-one with points in projective space and not with states in Hilbert space.} Hilbert space of pure states: if a state $|\psi\rangle\in \mathbb{C}^{d}$ is represented by a constellation $C$, then the rotated state $D(R)|\psi\rangle$ is represented by the constellation $R(C)$, where $R$ is a rigid rotation on three-dimensional space and $D(R)$ denotes its representation by a $\mathrm{SU}(2)$-matrix (in the corresponding spin-$s$ representation, with $d=2s+1$) acting on  qudit states. This characteristic of Majorana representation is relevant in the study of rotationally invariant properties of quantum states, such as geometric entanglement \cite{Aul:10,Mar:11} or spin-anticoherence \cite{Zim:06}. Recent surveys of applications of the Majorana representation can be found in \cite{LL2026, Kam:26}. In contrast to the Majorana representation, in the Bloch representation, the action of rotations for $d>2$ changes the Bloch vectors in a complicated way. Nevertheless, other properties, like some quantum correlation measures, are readily expressed in these terms ---provided the representation is properly extended to composite systems---see, e.g., \cite{Dev:08, Lou.Kul:21} and references therein. The choice of representation depends on the particular task in which we are interested, so in order to take advantage of both of them it is necessary to analyze the relation between them.

\section{Relations between representations}
\label{Section3}

A direct transformation between these representations would let us exploit the advantages of both. In this section, we express the Bloch vector of any pure qudit state in terms of its Majorana constellation, and comment on the difficulty of the inverse problem.

\subsection{From Majorana to Bloch representation}

In order to find the Bloch vector in terms of the Majorana stars, first we express the density matrix of a pure qudit state in terms of its Majorana constellation.
\begin{proposition}
\label{proposition_densityM}
The general expression of the density matrix $\rho=|\psi\rangle\langle \psi|$ of a pure qudit state $|\psi\rangle\in\mathbb{C}^d$ whose Majorana constellation is given by $d-1$ points $\hat{n}_r\in \mathcal{S}^2\subset\mathbb{R}^3$ with spherical coordinates $(\theta_r,\phi_r),\ \theta_r\in[0,\pi],\ \phi_r\in[0,2\pi)\ ,\ r=1,\dots,d-1$, has the form 
\begin{equation}
\label{qudit_state}
    \rho=\frac{1}{N^2}\sum_{m,k=1}^{d}\frac{(-1)^{k+m}}{(d-1)!}\sqrt{(d-m)!(d-k)!(m-1)!(k-1)!}W_{d-m}^{(d-1)}(W_{d-k}^{(d-1)})^*|m\rangle\langle k|,
\end{equation}
%\begin{equation}
%\label{qudit_state}
%    \rho=\frac{1}{N^2}\sum_{m,k=1}^{d}\frac{(-1)^{k+m}}{(d-1)!}\sqrt{(d-m)!(d-k)!(m-1)!(k-1)!}W_{d-m}^{(d-1)}(W_{d-k}^{(d-1)})^*|m\rangle\langle k|,
%\end{equation}

\noindent where $*$ denotes complex conjugation and for $0\leq n\leq d-1$, $W_n^{(t)}$ denotes the $n$-th fundamental symmetric function in the $t=d-1$ variables $\left\lbrace \zeta_{\hat{n}_r}\right\rbrace_{r=1}^{t}$,  namely
\begin{equation}
\label{wn}
    W_n^{(t)}=(-1)^{t-n}\sum_{q_1+\cdots+q_{d-1}=t-n} \zeta_{\hat{n}_1}^{q_1}\cdots \zeta_{\hat{n}_{t}}^{q_{t}},\ q_i\in\left\lbrace 0,1 \right\rbrace,\ \forall i=1,2\dots,t,
\end{equation}
and $\zeta_{\hat{n}_r}:=\tan \frac{\theta_r}{2}e^{i\phi_r},\ \forall r=1,2\dots,t,$ are the stereographic projections, from the south pole, of the Majorana stars to the complex plane. 
The normalization factor $N$, which also depends on the Majorana constellation, is
\begin{equation}
\label{normalization}
    N^2:=\sum_{r=1}^d\binom{d-1}{r-1}^{-1}|W_{d-r}^{(d-1)}|^2\ .
\end{equation}
\end{proposition}

\begin{proof}
Consider the stereographic projections $\left\lbrace \zeta_{\hat{n}_r}\right\rbrace_{r=1}^{d-1}$ of the Majorana stars $\left\lbrace \hat{n}_r\right\rbrace_{r=1}^{d-1}$, given by $\zeta_{\hat{n}_r}:=\tan \frac{\theta_r}{2}e^{i\phi_r},\ \forall r=1\dots,d-1$. The Majorana polynomial of the state $|\psi\rangle$ has $\zeta_{\hat{n}_r}$ as its roots so it has to be proportional to the polynomial $\prod_{r=1}^{d-1}(\zeta-\zeta_{\hat{n}_r})$, which by Vieta's formula can be written as $\sum_{r=0}^{d-1} W_r^{(d-1)}\zeta^{r}$, where $W_r^{(d-1)}$ are given as in \eqref{wn}. Comparing this polynomial, term by term in $\zeta$, with the Majorana polynomial of $|\psi\rangle$ as given in \eqref{majo}, we found that the components $\left\lbrace \psi^{(r)}\right\rbrace_{r=1}^{d}$ of the state vector $|\psi\rangle\in\mathbb{C}^d$ are
\begin{equation}
\label{comp_Psi_r}
    \psi^{(r)}=\frac{1}{N} \frac{(-1)^rW_{d-r}^{(d-1)}}{\sqrt{\binom{d-1}{r-1}}}=\frac{1}{N}(-1)^r\sqrt{\frac{(d-r)!(r-1)!}{(d-1)!}} W_{d-r}^{(d-1)}\ ,
\end{equation}
where  the normalization factor is
\begin{equation}
\label{normalizationState}
    N=\sqrt{\sum_{r=1}^d\frac{(d-r)!(r-1)!}{(d-1)!}|W_{d-r}^{(d-1)}|^2}\ .
\end{equation}
By \eqref{comp_Psi_r} and \eqref{normalizationState} in the definition of the density matrix $\rho=|\psi\rangle\langle\psi|$, Eqs. \eqref{qudit_state} and \eqref{normalization} follow.
\end{proof}

With this result we can derive a general expression for the Bloch vectors in terms of the Majorana constellation.
 
 \begin{theorem}
 \label{Theorem1}
     The Bloch vector \eqref{bloch_components} of an arbitrary pure qudit state $\rho\in\mathcal{L}_{+,1}(\mathbb{C}^d)$ as in \eqref{rho_gmrep} in terms of the spherical coordinates $(\theta_k,\phi_k) ,\ \theta_k\in[0,\pi],\ \phi_k\in[0,2\pi),\ k=1,\dots,d-1$, of its Majorana stars $\hat{n}_k=(\sin\theta_k\cos \phi_k,\sin\theta_k\sin \phi_k,\cos \theta_k)\in\mathcal{S}^2\subset\mathbb{R}^3$ is, following the notation established in Eqs. \eqref{r_stb1} and \eqref{r_stb2}, given by

 \begin{equation}
 \label{generalr}
     r_{\rho}=
\left(r_{12}^{(\mathrm{s})}, \ldots, r_{1 d}^{(\mathrm{s})}, \ldots, r_{d-1, d}^{(\mathrm{s})}, r_{12}^{(\mathrm{as})}, \ldots, r_{1 d}^{(\mathrm{as})}, \ldots, r_{d-1, d}^{(\mathrm{as})}, r_1^{(\mathrm{diag})}, \ldots, r_{d-1}^{(\mathrm{diag})}\right),
\end{equation}
where 
\small{
\begin{equation}
\label{parte_sim}
    r_{mk}^{(\mathrm{s})}=\frac{1}{N^2}\sqrt{\frac{2d}{d-1}}\frac{(-1)^{m+k}}{(d-1)!}\sqrt{(d-m)!(d-k)!(m-1)!(k-1)!}\mathrm{Re}(W_{d-m}^{(d-1)}(W_{d-k}^{(d-1)})^{*}),
\end{equation}
\begin{equation}
    r_{mk}^{(\mathrm{as})}=-\frac{1}{N^2}\sqrt{\frac{2d}{d-1}}\frac{(-1)^{m+k}}{(d-1)!}\sqrt{(d-m)!(d-k)!(m-1)!(k-1)!}\mathrm{Im}(W_{d-m}^{(d-1)}(W_{d-k}^{(d-1)})^{*}),
\end{equation}
and
\begin{equation}
\label{diagonalr}
    r_{l}^{(\mathrm{diag})}=\frac{1}{N^2}\frac{1}{(d-1)!}\sqrt{\frac{d}{(d-1)l(l+1)}}
    \left(
    -l(d-l-1)!\phantom{,}l!\phantom{,}|W_{d-l-1}^{(d-1)}|^2+\sum_{j=1}^l(d-j)!(j-1)!|W_{d-j}^{(d-1)}|^2
    \right).
\end{equation}}
\end{theorem}

\begin{proof}
     By Proposition \ref{proposition_densityM} and Eqs. \eqref{lambda_sym}-\eqref{lambda_diag} applied in the general formula for the Bloch vector \eqref{bloch_components}, with components \eqref{r_stb2}, the results in \eqref{parte_sim}-\eqref{diagonalr} follow.
\end{proof}
\subsection{Majorana constellation from the Bloch vector}

Due to the Abel-Ruffini theorem, we do not expect a general explicit formula for the Majorana stars in terms of the Bloch vector of a given density matrix for $d>5$. This is because to obtain them we need to find the roots of the polynomial given in \eqref{majo} with $\psi_m$ obtained by taking any nonzero column of $\rho$ (appropriately normalizing it). 

For small dimensions analytical expressions can be derived, but even in these cases they turn out to be rather cumbersome. Let us exemplify this in the case of $d=3$. We will now consider a qutrit state $\rho\in\mathcal{L}_{+,1}(\mathbb{C}^3)$ given in the form \eqref{rho_gmrep}, described by a vector $r_{\rho}\in\mathbb{R}^8$, and we will find its Majorana points in terms of $r_{\rho}$.

We will first find a quantum state $|\psi\rangle=(\psi^{(1)},\psi^{(2)},\psi^{(3)})^T$ such that $\rho=|\psi\rangle \langle \psi|$; to choose a concrete solution we will fix that $\psi^{(3)}$ is real and positive (this is always possible by multiplying the state by a global phase factor). Comparing a qutrit density matrix in representation \eqref{rho_gmrep} with a general qutrit density matrix for a pure state $\rho$ with entries $\rho^{(ij)}=\psi^{(i)}\left(\psi^{(j)}\right)^*$, we obtain that
\begin{align}
    |\psi^{(1)}|^2&= \frac{1}{3}+\frac{1}{\sqrt{3}}r^{(7)}+\frac{1}{3}r^{(8)}\ , \\
     \psi^{(1)} \left(\psi^{(2)}\right)^*&= (r^{(1)}-ir^{(4)})/\sqrt{3}\ , \\
      \psi^{(1)} \left(\psi^{(3)}\right)^*&= (r^{(2)}-ir^{(5)})/\sqrt{3} \ , \\
       |\psi^{(2)}|^2&= \frac{1}{3}-\frac{1}{\sqrt{3}}r^{(7)}+\frac{1}{3}r^{(8)}\ , \\
       \psi^{(2)} \left(\psi^{(3)}\right)^*&=  (r^{(3)}-ir^{(6)})/\sqrt{3}\ , \\
         |\psi^{(3)}|^2&=\frac{1}{3}-\frac{2}{3}r^{(8)} \ . 
\end{align}

From these equations we can easily find the components of the pure state in the computational basis, for example, by taking the third column of the density matrix divided by $\left( \psi^{(3)} \right)^*$ as the following unnormalized state (if $r^{(8)}\neq 1/2$); since there is a global phase freedom we can define that the third component of the state vector in $\mathbb{C}^3$ is real and positive, so $\psi^{(3)}=\sqrt{(1-2r^{(8)})/3}$ and the resulting state is:
\begin{equation}
\label{comp}
    \psi^{(1)}=\frac{r^{(2)}-ir^{(5)}}{\sqrt{1-2r^{(8)}}}, \ \ \ \psi^{(2)}=\frac{r^{(3)}-ir^{(6)}}{\sqrt{1-2r^{(8)}}}, \ \ \ \psi^{(3)}=\sqrt{\frac{1-2r^{(8)}}{3}}\ .
\end{equation}

For the state $|\psi\rangle$ with the above components \eqref{comp}, the Majorana polynomial \eqref{majo} takes the form
\begin{align}
\label{roots_r}
    P_{|\psi\rangle}(\zeta)&=-\psi^{(3)}+\sqrt{2}\psi^{(2)}\zeta-\psi^{(1)}\zeta^2 \nonumber\\
    &=\frac{1}{\sqrt{1-2r^{(8)}}}
    \left(-(r^{(2)}-ir^{(5)})\zeta^2+\sqrt{2}(r^{(3)}-ir^{(6)})\zeta-\frac{1-2r^{(8)}}{\sqrt{3}}
    \right),
\end{align}
with roots, for $r^{(2)}$ and $r^{(5)}$ non simultaneously equal to zero,  
\begin{equation}
\label{roots}
    \zeta_{1,2}=\frac{1}{\sqrt{2}(r^{(2)}-ir^{(5)})} \left(
    r^{(3)}-ir^{(6)}\pm\sqrt{(r^{(3)}-ir^{(6)})^2-\frac{2}{\sqrt{3}}(1-2r^{(8)})(r^{(2)}-ir^{(5)})}
    \right).
\end{equation}
The Majorana constellation, in this case consisting of two points, can be found by projecting the roots \eqref{roots} to the unit sphere by using the inverse stereographic projection from the south pole. In the cut $r^{(2)}=r^{(5)}=0$ there is always a star at the south pole. Even in this simple case, we can see that the resulting expressions turn complicated. 

\section{Relations between representations for pure permutation-symmetric multiqubit states}
\label{Section4}

The Majorana representation, originally defined for pure qudit states, can be straightforwardly extended to the case of multiqubit states with a symmetric state-vector due to the isomorphism between $n$-level pure states Hilbert spaces and those of permutation-symmetric $n$-qubit pure states. This idea has been used to analyze different properties of multipartite systems in terms of the standard Majorana representation in a large range of applications, see, e.g., \cite{Dev.Sud.Rah:12, Gol.Gra.etal:22,Mar.Gir.etal:10,Aul:10,Aul.Mar.Mur:10,Mar:11,Aul:12,Ser.Mir:23} and references therein; for a recent review of symmetric states, discussing the Majorana representation for their study, see \cite{Mar.Mul.etal:26}.

In this section, we will analyze how, in the case of pure two- and three-qubit permutation-symmetric states, both representations are related. Namely, we will find explicit analytical expressions of the basic elements describing correlation properties, defined in the Bloch representation, in terms of the corresponding Majorana stars (see Theorems \ref{Theorem2} and \ref{Theorem3} below).

\subsection{Two-qubit states}
A two-qubit pure state 
\begin{equation}
\label{2qstate}
|\psi_{2}\rangle=\psi^{(00)} |00\rangle+\psi^{(01)}  |01\rangle+\psi^{(10)}  |10\rangle+\psi^{(11)}  |11\rangle\in \mathbb{C}^2\otimes\mathbb{C}^2
\end{equation}
is permutation-symmetric when $\psi^{(01)}=\psi^{(10)} $; in this case it can be mapped into a three-level state by the correspondence $|00\rangle \rightarrow |1\rangle$, $(|01\rangle+|10\rangle)/\sqrt{2} \rightarrow |2\rangle$ and $|11\rangle \rightarrow |3\rangle$. Then, the components $\psi^{(1)} ,\psi^{(2)}, \psi^{(3)}$ of the corresponding three-level state in terms of the components of (\ref{2qstate}) are $\psi^{(1)}=\psi^{(00)}, \psi^{(2)}=\sqrt{2}\psi^{(01)}$ and $\psi^{(3)}=\psi^{(11)} $. Then, by \eqref{majo}, the points that represent the state (\ref{2qstate}) on the unit sphere are the stereographic projections, from the south pole, of the roots of the Majorana polynomial
\begin{equation}
    P_{|\psi_2\rangle}(\zeta)=-\psi^{(00)} \zeta^2+2\psi^{(01)} \zeta-\psi^{(11)}\ .
\end{equation}

The behavior of the Majorana constellation under the action of rotations is similar to that in the one-particle case: when applying a physical rotation, the new state in the projective space $P(\mathbb{C}^2\otimes\mathbb{C}^2)$ has a constellation which is the original one rotated in three-dimensional space by the same angle and axis as the physical rotation, since a symmetrized state of $s$ spin-$1/2$ states is equivalent, under the action of the rotation group, to a single spin-$s$ state; here we have exactly the same situation, but not constrained to the particular description of states of spin systems.

In terms of the Pauli representation, the particular case for $d=2$ of the representation given in \eqref{2qudits}, a two-qubit state density matrix is expressed as
\begin{equation}
\label{rho_two_qubits}
    \rho=\frac{1}{4}\left(
    \mathbb{I}_{\mathbb{C}^2}\otimes \mathbb{I}_{\mathbb{C}^2}+r_{\rho}\cdot\sigma\otimes \mathbb{I}_{\mathbb{C}^2}+\mathbb{I}_{\mathbb{C}^2} \otimes s_{\rho}\cdot\sigma+\sum_{i,j=1}^3T_{\rho}^{(ij)}\sigma^{(i)}\otimes \sigma^{(j)}
    \right),  
\end{equation}
where $\sigma=(\sigma^{(1)},\sigma^{(2)},\sigma^{(3)})$; $\sigma^{(i)},\ i=1,2,3$, are the Pauli operators; $r_{\rho}^{(i)}=\mathrm{tr}(\rho\ \sigma^{(i)}\otimes\mathbb{I}_{\mathbb{C}^2})$, $s_{\rho}^{(i)}=\mathrm{tr}(\rho\ \mathbb{I}_{\mathbb{C}^2}\otimes\sigma^{(i)})$ and $T_{\rho}^{(ij)}=\mathrm{tr}(\rho\  \sigma^{(i)}\otimes \sigma^{(j)})$. The parametrization in \eqref{rho_two_qubits} is commonly referred to as the \emph{Fano form} \cite{Fan:83} of the density matrix.

The purity condition $\rho^2=\rho$ implies that: 
\begin{equation}
\label{purity1}
\|r_{\rho}\|^2+\|s_{\rho}\|^2+\sum_{i,j=1}^3(T_{\rho}^{(ij)})^2=3,\  r_{\rho}^{(c)}=\sum_{j=1}^3s_{\rho}^{(j)}T_{\rho}^{(cj)}, \   s_{\rho}^{(c)}=\sum_{i=1}^3r_{\rho}^{(i)}T_{\rho}^{(ic)}\ . 
\end{equation}
If the state is permutation-symmetric then $r_{\rho}=s_{\rho}$ and the correlation matrix $T_{\rho}$ is a symmetric matrix.

A pure symmetric $2$-qubit state $|\psi_2\rangle$ can be represented by two Majorana stars on the unit sphere. We can write the state in terms of the spherical coordinates of this two points $\hat{n}_1(\theta_1,\phi_1)$ and $\hat{n}_2(\theta_2,\phi_2)$, namely,
\begin{equation}
\label{2qsymm_stars}
|\psi_2\rangle = \frac{|\psi_{\hat{n}_1}\rangle\otimes |\psi_{\hat{n}_2}\rangle+|\psi_{\hat{n}_2}\rangle\otimes |\psi_{\hat{n}_1}\rangle}{\sqrt{2(1+|\langle \psi_{\hat{n}_1}|\psi_{\hat{n}_2}\rangle|^2)}}\ ,
\end{equation}
where, for $k=1,2$, 
\begin{equation}
\label{ind}
    |\psi_{\hat{n}_k}\rangle=\cos \frac{\theta_k}{2}|0\rangle+e^{i \phi_k}\sin\frac{\theta_k}{2}|1\rangle
\end{equation}
are the single-qubit states with a star in the direction $\hat{n}_k,\ k=1,2$.

Then, we can derive the following result for the symmetrized state that has this two points $\hat{n}_1(\theta_1,\phi_1)$ and $\hat{n}_2(\theta_2,\phi_2)$ as the stars forming its Majorana constellation.

 \begin{theorem}
 \label{Theorem2}
 Let $\rho_2=|\psi_2\rangle\langle\psi_2|\in\mathcal{L}_{+,1}(\mathbb{C}^2\otimes\mathbb{C}^2)$ be the density matrix of a pure state with permutation-symmetric vector $|\psi_2\rangle = \left(|\psi_{\hat{n}_1}\rangle\otimes |\psi_{\hat{n}_2}\rangle+|\psi_{\hat{n}_2}\rangle\otimes |\psi_{\hat{n}_1}\rangle\right)/\sqrt{2(1+|\langle \psi_{\hat{n}_1}|\psi_{\hat{n}_2}\rangle|^2)}$.
 Then, the Bloch vector of the reduced state $r_{\rho_2}=s_{\rho_2}$ and the correlation matrix $T_{\rho_2}$ have the form
 \begin{equation}
 \label{rT_two_qubits}
     r_{\rho_2}=\frac{2(\hat{n}_1+\hat{n}_2)}{3+\hat{n}_1\cdot\hat{n}_2}\ , \ \ T_{\rho_2}=\frac{(1-\hat{n}_1\cdot \hat{n}_2)\mathbb{I}_{\mathbb{C}^3}+2\hat{n}_1\boldsymbol{\vee} \hat{n}_2}{3+\hat{n}_1\cdot\hat{n}_2}\ ,
 \end{equation}
 where $a \vee b:=ab^T+ba^T$ denotes the symmetrized tensor product\footnote{See \eqref{gen-vee} for a general definition, note that here we are not using the standard normalization dividing by the cardinality of the permutation group.} of two vectors in $\mathbb{R}^3$, i. e., the matrix with entries $(a \vee b)^{(kl)}:=a^{(k)}b^{(l)}+b^{(k)}a^{(l)}$, and $\hat{n}_1,\hat{n}_2$, are the Majorana stars of the corresponding three-level state obtained by the correspondence $|00\rangle \rightarrow |1\rangle$, $(|01\rangle+|10\rangle)/\sqrt{2} \rightarrow |2\rangle$ and $|11\rangle \rightarrow |3\rangle$.
 \end{theorem}
\begin{proof} By \eqref{ind} and the above definition \eqref{2qsymm_stars} of $|\psi_2\rangle$, we found
    \begin{align}
    |\psi_2\rangle &= \frac{1}{\sqrt{2(1+|\langle \psi_{\hat{n}_1}|\psi_{\hat{n}_2}\rangle|^2)}}\left(|\psi_{\hat{n}_1}\rangle\otimes |\psi_{\hat{n}_2}\rangle+|\psi_{\hat{n}_2}\rangle\otimes |\psi_{\hat{n}_1}\rangle\right)\\
    & =\frac{1}{N}\left(2\cos \frac{\theta_1}{2}\cos \frac{\theta_2}{2}|00\rangle+
   2e^{i(\phi_1+\phi_2)}\sin\frac{\theta_1}{2} \sin\frac{\theta_2}{2}|11\rangle \right. \nonumber\\
   &+\left.
   \left(
    \cos \frac{\theta_1}{2}\sin \frac{\theta_2}{2}e^{i\phi_2}+\cos \frac{\theta_2}{2}\sin \frac{\theta_1}{2}e^{i\phi_1}
    \right)(|01\rangle+|10\rangle)\right)\ ,\nonumber
\end{align}
where
\begin{align}
    N&=\sqrt{2(1+|\langle \psi_{\hat{n}_1}|\psi_{\hat{n}_2}\rangle|^2)}\nonumber\\
    &=\sqrt{2\left(1+\left\lvert \cos\frac{\theta_1}{2}\cos\frac{\theta_2}{2}+e^{i(\phi_2-\phi_1)}\sin\frac{\theta_1}{2} \sin\frac{\theta_2}{2}\right\rvert^2\right)}\nonumber\\
    &=\sqrt{3+\sin \theta_1 \sin \theta_2 \sin \phi_1 \sin\phi_2+\sin\theta_1 \sin \theta_2 \cos \phi_1\cos\phi_2+\cos\theta_1 \cos
  \theta_2}\\&=\sqrt{3+\hat{n}_1\cdot\hat{n}_2}\ .\nonumber
\end{align}
Computing the corresponding Bloch vector $r_{\rho_2}=s_{\rho_2}$  of the state $\rho_2=  |\psi_2\rangle\langle \psi_2|$, we get
\begin{align}
  r_{\rho_2}&=\sum_{i=1}^3\mathrm{tr}(\rho_2 \sigma^{(i)}\otimes\mathbb{I}_{\mathbb{C}^2})\hat{e}_i=\sum_{i=1}^3\langle \psi_2|\sigma^{(i)}\otimes\mathbb{I}_{\mathbb{C}^2}|\psi_2\rangle\hat{e}_i\nonumber\\
  &=2
(\sin\theta_1\cos\phi_1+\sin\theta_2\cos\phi_2,\sin\theta_1\sin\phi_1+\sin\theta_2\sin\phi_2,\cos\theta_1+\cos\theta_2)/(3+\hat{n}_1\cdot \hat{n}_2)\ , \\ &=\frac{2(\hat{n}_1+\hat{n}_2)}{3+\hat{n}_1\cdot \hat{n}_2}\ .\nonumber
\end{align}
For the state $\rho_2$,
the diagonal elements $ T_{\rho_2}^{(ii)}=\langle \psi_2|\sigma^{(i)}\otimes\sigma^{(i)}|\psi_2\rangle, i=1,2,3$, of its correlation matrix $T_{\rho_2}$  are
\small\begin{align}
\label{t_diag}
    T_{\rho_2}^{(ii)}&=\frac{1}{N}\left(1
    -(1-4\delta_{1i})\sin\theta_1\sin\theta_2\cos\phi_1\cos\phi_2 -(1-4\delta_{2i})\sin\theta_1\sin\theta_2\sin\phi_1\sin\phi_2\right. \\ & \left. -(1-4\delta_{3i})\cos\theta_1\cos\theta_2
    \right),\nonumber
\end{align}
\normalsize and the other three independent entries are
\begin{align}
\label{off_diagonal}
    T_{\rho_2}^{(12)}&=2(\sin \theta_1 \sin \theta_2 \sin\phi_1 \cos\phi_2+\sin \theta_1\sin \theta_2\cos\phi_1\sin\phi_2)/N\  ,\nonumber\\
     T_{\rho_2}^{(13)}&=2(\sin \theta_1 \cos\theta_2\cos \phi_1 +\cos \theta_1 \sin\theta_2 \cos\phi_2)/N\ ,\\
      T_{\rho_2}^{(23)}&=2(\cos\theta_2\sin\theta_1\sin\phi_1+\cos\theta_1\sin\theta_2\sin\phi_2)/N\ .\nonumber
\end{align}
On the other hand, a straightforward computation gives
\small{
\begin{equation}
\label{n1veen2}
\hat{n}_1\vee\phantom{.} \hat{n}_2=
\left(
\begin{array}{ccc}
 2 s\theta _1 s\theta _2 c\phi _1 c \phi _2 & s \theta _1 s \theta _2 s \phi _1 c \phi _2+s
  \theta_1 s \theta _2 s \phi_2 c \phi_1 & s \theta_1 c \theta _2 c \phi_1 + s \theta_2 c
 \theta_1 c\phi_2 \\
s\theta_1 s\theta_2 s\phi_1 c \phi_2+ s\theta_1 s \theta _2 s\phi_2 c\phi _1 & 2 s \theta _1 s\theta _2 s\phi_1 s\phi _2 &  s\theta_1 c\theta_2 s\phi_1+ s\theta_2 c\theta_1  s\phi_2\\
s\theta _1c\theta _2c\phi _1+s\theta _2c\theta _1c\phi _2& s \theta _1c\theta _2s\phi _1+s\theta _2c\theta_1 s\phi_2 & 2 c \theta_1 c\theta_2 \\
\end{array}
\right),
\end{equation}}
where $s\alpha, c\alpha$ denote $\sin\alpha, \cos\alpha$, respectively. The off-diagonal entries of \eqref{n1veen2} coincide with the off-diagonal entries of the correlation matrix $T_{\rho_2}$ in \eqref{off_diagonal}, when dividing them by $\frac{N}{2}=\sqrt{3+\hat{n}_1\cdot\hat{n}_2}/2$. Adding to diagonal entries of \eqref{n1veen2}, multiplied by 2, the term $(1-\hat{n}_1\cdot\hat{n}_2)/2=(1-s\theta_1s\theta_2(c\phi_1 c\phi_2+s\phi_1 s\phi_2)-c\theta_1 c\theta_2)/2$ and dividing by $N=\sqrt{3+\hat{n}_1\cdot\hat{n}_2}$, we found the diagonal entries of the correlation matrix in \eqref{t_diag},  therefore \eqref{rT_two_qubits} follows.
\end{proof}

This theorem then allows us to have a geometrical interpretation of the direction of the vector sum of the directions of the Majorana stars of the constituent qubits of a permutation-symmetric two-qubit state as the direction of the corresponding Bloch vector in the Bloch sphere representation. For coherent states ($\hat{n}_1=\hat{n}_2=:\hat{n}$) we have $r=\hat{n}$ and $T^{(ij)}=\hat{n}^{(i)}\hat{n}^{(j)}$, corresponding to separable states; and, in the anticoherent case ($\hat{n}_1=-\hat{n}_2$), the Bloch vector vanishes, so corresponding to maximally entangled states, the correlation matrix takes the form $T^{(ij)}=\delta_{ij}-2\hat{n}^{(i)}\hat{n}^{(j)}$.

\subsection{Three-qubit states}

In a similar way to the two-qubit case, it is possible to express the Bloch representation parameters in terms of the Majorana stars for a permutation-symmetric pure state of three qubits. For pure symmetric three-qubit states we have a correspondence with four-level states $|000\rangle\rightarrow|1\rangle$, $(|001\rangle+|010\rangle+|100\rangle)/\sqrt{3}\rightarrow|2\rangle$, $(|011\rangle+|101\rangle+|110\rangle)/\sqrt{3}\rightarrow|3\rangle$ and $|111\rangle\rightarrow|4\rangle$, it is easy to note that the Majorana stars representing a symmetric three-qubit state $|\psi_{3}\rangle=\sum_{i,j,k=0}^1 \psi^{(ijk)}|ijk\rangle$ in the standard basis, where $\psi^{(ijk)}$ are fully symmetric coefficients under permutations of indices, is given by the stereographically projected roots of the polynomial
\begin{equation}
    P_{|\psi_3\rangle}(\zeta)=-\psi^{(000)}\zeta^3+3\psi^{(001)}\zeta^2-3\psi^{(011)}\zeta+\psi^{(111)}\ .
\end{equation}

On the other hand, in the Bloch representation a qutrit state density matrix $\rho\in\mathcal{L}_{+,1}(\mathbb{C}^2\otimes \mathbb{C}^2 \otimes \mathbb{C}^2)$ can be written as
\begin{align}
    \rho&=\frac{1}{8}
    \left(
\mathbb{I}_{\mathbb{C}^2}\otimes\mathbb{I}_{\mathbb{C}^2}\otimes \mathbb{I}_{\mathbb{C}^2}
  +r_1\cdot\sigma\otimes\mathbb{I}_{\mathbb{C}^2}\otimes \mathbb{I}_{\mathbb{C}^2}
   +\mathbb{I}_{\mathbb{C}^2}\otimes r_2\cdot\sigma\otimes\mathbb{I}_{\mathbb{C}^2}
    +\mathbb{I}_{\mathbb{C}^2}\otimes \mathbb{I}_{\mathbb{C}^2}\otimes r_3\cdot\sigma \right. \\
   &+\sum_{ij}T_{12}^{(ij)} \sigma^{(i)}\otimes\sigma^{(j)}\otimes \mathbb{I}_{\mathbb{C}^2} 
   +\sum_{ij}T_{13}^{(ij)} \sigma^{(i)}\otimes \mathbb{I}_{\mathbb{C}^2} \otimes \sigma^{(j)}
   +\sum_{ij}T_{23}^{(ij)} \mathbb{I}_{\mathbb{C}^2}\otimes\sigma^{(i)}\otimes\sigma^{(j)}\nonumber \\
&+\sum_{ijk}M_{123}^{(ijk)}\sigma^{(i)}\otimes\sigma^{(j)}\otimes \sigma^{(k)}\left.
    \right) ,\nonumber
\end{align}

\noindent where, for $\rho=\sum_{i,j,k}\rho^{(ijk)}|ijk\rangle\langle ijk|$, all the components are defined as $\rho^{(ijk)}:=\mathrm{tr}(\rho X_{i}\otimes X_{j}\otimes X_{k})$, with $X_{0}=\mathbb{I}_{\mathbb{C}^2}$ and $X_{i}:=\sigma^{(i)}$, being $\sigma^{(i)}$ the Pauli operators, $i=1,2,3$, and defining, for all $i,j,k\in \left\lbrace 1,2,3\right\rbrace$, the Bloch vectors $r_1^{(i)}:=\rho^{(i00)},\ r_2^{(i)}:=\rho^{(0i0)},\ r_3^{(i)}:=\rho^{(00i)}$, the correlation matrices on the corresponding subspaces $T_{12}^{(ij)}:=\rho^{(ij0)},\  T_{13}^{(ij)}:=\rho^{(i0j)},\  T_{23}^{(ij)}:=\rho^{(0ij)}$, and the third-order correlation tensor $M_{123}^{(ijk)}:=\rho^{(ijk)}$.

Given a (normalized) permutation-symmetric three-qubit pure  state corresponding to a four-level state with Majorana constellation $\left\lbrace \hat{n}_1,\hat{n}_2,\hat{n}_3\right\rbrace$,
\begin{equation}
\label{psi3}
    |\psi_3\rangle =\frac{1}{\sqrt{18(1+I)}}\sum_{\tau\in S_3} |\psi_{\hat{n}_{\tau(1)}}\rangle\otimes|\psi_{\hat{n}_{\tau(2)}}\rangle\otimes |\psi_{\hat{n}_{\tau(3)}}\rangle\ ,
\end{equation}
where the sum is taken over the set of permutations of three elements $S_3$ and $ |\psi_{\hat{n}_p}\rangle=\cos \frac{\theta_p}{2}|+\rangle+e^{i \phi_p}\sin\frac{\theta_p}{2}|-\rangle , p=1,2,3$, and $I:=(\hat{n}_1\cdot \hat{n}_2+\hat{n}_1\cdot \hat{n}_3+\hat{n}_2\cdot \hat{n}_3)/3$. We can express the above defined parameters $\rho^{(ijk)}$ of its density matrix $\rho_3=|\psi_3\rangle\langle \psi_3|$ in terms of its Majorana stars $\hat{n}_p,\ p=1,2,3$. For this class of states $r_1=r_2=r_3=:r_{\rho_3}$, $T_{12}=T_{13}=T_{23}=:T_{\rho_3}$ due to the symmetry, and the tensor $M_{\rho_3}$, with components $M_{\rho_3}^{(ijk)}:=M_{123}^{(ijk)}$, is fully symmetric under permutations of any of the indices.

After computing all the needed traces of the density matrix of the state \eqref{psi3}, analogously to the two-qubit case, we identified the following structure of the correlation parameters of a pure permutation-symmetric three-qubit state in terms of its Majorana constellation given by the aforementioned isomorphism. We summarize our findings in Theorem \ref{Theorem3}.

\begin{theorem}
\label{Theorem3}
 Let $\rho_3=|\psi_3\rangle\langle\psi_3|\in\mathcal{L}_{+,1}(\mathbb{C}^2\otimes\mathbb{C}^2\otimes\mathbb{C}^2)$ be the density matrix of pure state with permutation-symmetric vector $|\psi_3\rangle$ of the form \eqref{psi3}.
 Then, the corresponding correlation parameters $r_{\rho_3},T_{\rho_3},M_{\rho_3}$, of tensor order 1, 2, and 3, respectively,  can be written in terms of its Majorana stars $\hat{n}_p\in\mathcal{S}^2\subset\mathbb{R}^3,\ p=1,2,3$, as follows:
 \begin{align}
     r_{\rho_3}&=\frac{1}{9(1+I)}\left(6(\hat{n}_1+\hat{n}_2+\hat{n}_3)-\vec{v}\right),\\
     T_{\rho_3}&=\frac{1}{9(1+I)}\left(3(1-I)\mathbb{I}_{\mathbb{R}^3}+3\sum_{i<j}\hat{n}_i\vee\hat{n}_j\right),\\
     M_{\rho_3}&=\frac{1}{9(1+I)}\left(
   S_{\vec{v}}
   +  3\ \hat{n}_1\vee\hat{n}_2\vee\hat{n}_3  \right),
     \label{tensorM}
 \end{align}
 where $\vec{v}:=\hat{n}_1 (1-\hat{n}_2\cdot \hat{n}_3)+\hat{n}_2(1-\hat{n}_1\cdot \hat{n}_3)+\hat{n}_3(1-\hat{n}_1\cdot \hat{n}_2)$,  $I=(\hat{n}_1\cdot \hat{n}_2+\hat{n}_1\cdot \hat{n}_3+\hat{n}_2\cdot \hat{n}_3)/3$, and $S_{\vec{v}}$ 
 denotes the symmetric tensor $S_{\vec{v}}=\sum_{ijk}\left(\delta_{ij}v^{(k)}+ \delta_{ik}v^{(j)}+ \delta_{jk}v^{(i)}\right)\hat{e}_i\otimes\hat{e}_j\otimes\hat{e}_k$; $\hat{e}_i$, $i=1,2,3$, denotes the standard unitary basis vectors in $\mathbb{R}^3$.  Here the symmetric product of $r$ vectors $\hat{n}_1,\dots,\hat{n}_r$ is defined as
 \begin{equation}
     \label{gen-vee}\hat{n}_1\vee\hat{n}_2\vee\dots\vee\hat{n}_r:=\sum_{\tau\in S_r} \hat{n}_{\tau(1)}\otimes\hat{n}_{\tau(2)}\otimes\dots\otimes\hat{n}_{\tau(r)}.
 \end{equation}
 \end{theorem}

\begin{remark}
\label{Remark1}
    The components of tensor $M_{\rho_3}$ have the form 
    \begin{equation}
        M_{\rho_3}^{(ijk)}=\frac{1}{9(1+I)}\left( \delta_{ij}v^{(k)}+ \delta_{ik}v^{(j)}+ \delta_{jk}v^{(i)}+3\sum_{\tau\in S_3} \hat{n}_{\tau(1)}^{(i)} \hat{n}_{\tau(2)}^{(j)}  \hat{n}_{\tau(3)}^{(k)} \right).
    \end{equation}
    
     Due to the symmetry, $M_{\rho_3}$ has only $10$ independent components, namely:
    \begin{itemize}
        \item[(i)] three with all equal indices, which can be expressed as 
        \begin{equation}
            M^{(iii)}_{\rho_3}=\frac{1}{9(1+I)}\left(3v^{(i)}+18\hat{n}_1^{(i)} \hat{n}_2^{(i)} \hat{n}_3^{(i)}\right),
        \end{equation}
        \item[(ii)] six with two equal indices and one different, of the form
        \begin{equation}
            M^{(iij)}_{\rho_3}=\frac{1}{9(1+I)}\left(v^{(j)}+3 \sum_{\tau\in S_3}\hat{n}_{\tau(1)}^{(i)}\hat{n}_{\tau(2)}^{(i)}\hat{n}_{\tau(3)}^{(j)}\right) ,
        \end{equation}
    \item[(iii)] and the term with all different indices
        \begin{equation}
            M^{(123)}_{\rho_3}=\frac{1}{9(1+I)}\left(3\ \mathrm{perm}(\hat{n}_1,\hat{n}_2,\hat{n}_3)\right),
        \end{equation}
       \noindent where $\mathrm{perm}(u,v,w)$ denotes the permanent of the matrix with columns $u,v,w$.
    \end{itemize}
\end{remark}

\section{Quantum correlations in terms of the Majorana stars}
\label{Section5}
Our results in the previous section allow us to analyze quantum correlations of various classes of quantum states in terms of their Majorana constellations. This geometric approach could be useful, for example, in problems involving the optimization of rotationally invariant functions on quantum states, since in that case the rotational degrees of freedom  can be removed, i.e., we can optimize a function in the space of shapes of constellations regardless of its particular orientation in space.

As an example of the application of the theorems derived above, in this section we quantify quantum correlations, measured via the geometric quantum discord \cite{Dak.Ved.Bru:10}, for all possible configurations of stars of permutation-symmetric pure two-qubit states. This example also reveals a connection between the concepts of geometric quantum discord and the measure of $1$-anticoherence \cite{baguette2017anticoherence} based on total variance for these classes of states, when considering spin systems, see Theorem \ref{Theorem4}.

 The geometric quantum discord of a two-qudit state $\rho\in \mathcal{L}_{+,1}(\mathbb{C}^{d_{1}}\otimes\mathbb{C}^{d_{2}})$, denoted by $D_{G}(\rho)$, was introduced \cite{Dak.Ved.Bru:10,Luo.Fu:10} as the minimal (squared) distance between $\rho$ and the set $\Omega$ of states with vanishing quantum discord\footnote{In the bipartite case, the states for which quantum discord vanishes are identified with the classical-quantum states, which are states on $\mathcal{H}_A\otimes \mathcal{H}_B$ invariant by measurements on $A$, and have the form
\begin{equation}
    \rho=\sum_i p_i  |i\rangle \langle i | \otimes \rho_{B_i} \ ,
\end{equation}
where the vectors $\lbrace |i\rangle\rbrace$ constitute an orthonormal set on $\mathcal{H}_A$ and $p_i$ is a probability distribution ($p_i >0$ and $\sum_i p_i=1$). Analogously, quantum-classical states are defined, leading to, consequently, two different notions, left and right geometric quantum discord. For permutation-symmetric states this distinction is not relevant.},
\begin{equation}
D_{G}(\rho):=\min_{\chi\in\Omega}\|\rho-\chi\|_{2}^{2},
\label{vedral}%
\end{equation}
where $\|\cdot\|_2$ denotes the Hilbert-Schmidt norm.
Alternative definitions have been defined using other notions of distance between states \cite{Pau.Oli.Sar:13, Spe.Ors:13,Rog.Spe.Ill:16}.
The geometric quantum discord \eqref{vedral} of an arbitrary two-qubit state $\rho\in \mathcal{L}_{+,1}(\mathbb{C}^2 \otimes \mathbb{C}^2)$, is given by \cite{Dak.Ved.Bru:10} :
\begin{equation}
    \label{2qubitDiscordHS}
    D_{G}(\rho)=\frac{1}{4}
    \left(
    \mathrm{tr}G(\rho)-\lambda_{max}\left[G(\rho)\right]
    \right),
\end{equation}
where $G(\rho):=
r_{\rho}r_{\rho}^T+T_{\rho}T_{\rho}^T$, being $r_{\rho}^{(i)}=\mathrm{tr}[\rho \ \sigma^{(i)} \otimes \mathbb{I}_{\mathbb{C}^2}]$ and $T^{(ij)}_{\rho}=\mathrm{tr}[\rho\  \sigma^{(i)} \otimes \sigma^{(j)}]$, for $i,j=1,2,3$; $\lambda_{max}\left[G(\rho)\right]
   $ denotes the largest eigenvalue of $G(\rho)$.

Since the geometric quantum discord is invariant under local rotations, we can focus on the degrees of freedom of the states associated with the shape of the set of points on the unit sphere given by the Majorana representation, ignoring the particular orientation on the space of this configuration of points. At the geometric level, the space of states has the structure of a fiber bundle, where the base space corresponds to the space of shapes (which formally captures the possible shapes of constellations, independently of their orientation) and the fibers are homeomorphic to $\mathrm{SU}(2)$, carrying the information about the possible spatial orientations of a fixed shape.

Pure-symmetric two-qubit states are represented, under Majorana's prescription, as two points, possibly degenerated, on the unit sphere. States, up to a rotation, can be parametrized in the following way: consider the state $|\psi_{\alpha}\rangle$ whose Majorana representation is given by two stars in the $xz$-plane, symmetrically disposed with respect to $z$-axis, separated by an angle $\alpha\in [0,\pi]$; this means that the two Majorana stars are separated by an angle $\alpha$ and located at 
\begin{equation}
\label{points_example}
\hat{n}_1=(\sin(\alpha/2),0,\cos(\alpha/2)),\ \hat{n}_2=(-\sin(\alpha/2),0,\cos(\alpha/2))\ .
\end{equation}

By applying Theorem \ref{Theorem2} to the state with stars given by \eqref{points_example}, we find that the Bloch vector(s) and the correlation matrix of this state are
\begin{align}
    r=\left(0,0,\frac{4\cos(\alpha/2)}{3+\cos\alpha}\right), \ 
    T=\left(
\begin{array}{ccc}
 1-\frac{4}{\cos \alpha +3} & 0 & 0 \\
 0 & \frac{4}{\cos \alpha +3}-1 & 0 \\
 0 & 0 & 1 \\
\end{array}
\right),
\end{align}

then,  by evaluating (\ref{2qubitDiscordHS}) in $\rho_{\alpha}:=|\psi_{\alpha}\rangle \langle \psi_{\alpha} |$, we obtain that
\begin{equation}
\label{dg2}
    D_G(\rho_{\alpha})=\frac{2\sin^4\frac{\alpha}{2}}{(3+\cos\alpha)^2}
    \ .
\end{equation}

By the relation  $D_{G}(\rho)=C(\rho)^2/2$ valid \cite{Rog.Spe.Ill:16} for any pure bipartite state $\rho$, we can read the concurrence of the state in terms of the shape parameter $\alpha$:
\begin{equation}
    C(\rho_{\alpha})=\frac{2\sin^2 \frac{\alpha}{2}}{|3+\cos\alpha|}\ .
\end{equation}

The result in equation \eqref{dg2} is interesting since, when regarding the state as a spin-$1$ state as in \eqref{psialpha}, the geometric quantum discord exactly coincides with its value of the $1$-anticoherence measure based on total variance \cite{baguette2017anticoherence}, 
defined as
\begin{equation}
\label{1antic}
 A_1^{\mathbb{V}}(|\psi\rangle)=\frac{s^2-|\langle\vec{S}\rangle|^2}{s^2}\ .
\end{equation}

The expression on the RHS of \eqref{dg2} was found in \cite{baguette2017anticoherence} as an example of this measure of anticoherence. To see this note that, by \eqref{comp_Psi_r}, the state whose constellation is given by \eqref{points_example} is
\begin{equation}
\label{psialpha}
    |\tilde{\psi}_{\alpha}\rangle=\frac{2}{\sqrt{3+\cos \alpha}}
    \left(
    \cos^2 (\alpha/4) |1,1\rangle-\sin^2 (\alpha/4)|1,-1\rangle
    \right).
\end{equation}

Then, considering the spin-$1$ generators
\begin{equation}
    S_x=\frac{1}{\sqrt{2}}\begin{pmatrix}
        0&1&0\\
        1&0&1\\
        0&1&0\\
    \end{pmatrix}, \  \  \ 
    S_y=\frac{1}{\sqrt{2}i}\begin{pmatrix}
        0&1&0\\
        -1&0&1\\
        0&-1&0\\
    \end{pmatrix}, \  \  \
    S_z=\begin{pmatrix}
        1&0&0\\
        0&0&0\\
        0&0&-1\\
    \end{pmatrix}, 
\end{equation}
and computing the spin expectation value, we find that $\langle S_x \rangle_{|\tilde{\psi}_{\alpha}\rangle}=\langle S_y \rangle_{|\tilde{\psi}_{\alpha}\rangle}=0$ and
\begin{equation}
    \langle S_z \rangle_{|\tilde{\psi}_{\alpha}\rangle}=\frac{4 \cos \frac{\alpha }{2}}{3+\cos \alpha },
\end{equation}
thus, this implies, by definition \eqref{1antic} and the above result \eqref{dg2}, that
\begin{equation}
    A_1^{\mathbb{V}}(|\tilde{\psi}_{\alpha}\rangle)=\frac{4 \sin ^4\frac{\alpha }{2}}{(3+\cos \alpha)^2}=2D_{G}(|\psi_{\alpha}\rangle).
\end{equation}

Therefore, taking into account that this parametrization covers all possible permutation-symmetric states up to orientation, and both, geometric discord and 1-anticoherence are constant along the $\mathrm{SU}(2)$-fibers for this class of states, we have the following result.

\begin{theorem}
\label{Theorem4}
    The geometric quantum discord with respect the Hilbert-Schmidt distance of a pure-symmetric state of two-qubits coincides with half of the value of the 1-anticoherence of its corresponding spin vector under the correspondence $|00\rangle \rightarrow |1,1\rangle$, $(|01\rangle+|10\rangle)/\sqrt{2} \rightarrow |1,0\rangle$ and $|11\rangle \rightarrow |1,-1\rangle$.
\end{theorem}

The above result in Theorem \ref{Theorem4} confirms quantitatively the relation between spin $1$-anticoherence of a spin state and quantum correlations (in the form of concurrence or geometric quantum discord) of the corresponding symmetrized state in the case of two qubits. In \cite{baguette2017anticoherence}, a result of the same nature is found in case of a measure of $1$-anticoherence based on purity which is related to the linear entropy of the reduced state of the density matrix of the corresponding symmetrized state.

\section{Conclusions}
\label{Section6}
In this work, we compared two of the most commonly used extensions of the Bloch sphere representation of qubits to the case of states of arbitrary dimension $d>2$, finding explicit analytical relations between them. In particular, in Section \ref{Section3}, we found an analytical expression of the density matrix in terms of the Majorana constellation of the state (Proposition \ref{proposition_densityM}) and also of the Bloch vector in $\mathbb{R}^{d^2-1}$ in terms of the $d-1$ Majorana stars of an arbitrary qudit state (Theorem \ref{Theorem1}).

Both representations can be used in the case of multipartite  pure symmetric states of $n$ qubits. For $n=2,3$ cases we found, in Section \ref{Section4}, new explicit analytical relations that allows us to express: (i) the Bloch vectors and correlations matrix of a pure symmetric two-qubit state in terms of the two stars constituting its Majorana constellation (Theorem \ref{Theorem2}) and (ii) the Bloch vectors, correlations matrices and correlations tensor of a pure symmetric three-qubit state in terms of the three stars of its Majorana constellation (Theorem \ref{Theorem3}).

Theorems \ref{Theorem2} and \ref{Theorem3} imply that we can make use of the Majorana representation to analyze quantum correlations, usually described in terms of the aforementioned tensors of different ranks. In this vein, in Section \ref{Section5}, we compared the geometric quantum discord for pure symmetric states of two qubits, for all possible configurations of the Majorana stars, revealing a connection with the measure of $1$-anticoherence of spin states, as we summarize in Theorem \ref{Theorem4}. Beyond this application, these results  could also be useful when solving optimization problems of rotationally invariant functions on the set of pure symmetric quantum states in the two- and three-qubit cases, by reducing the number of required parameters to $1$ and $3$ parameters, respectively, since the optimization can be carried out in the shape space.

An important open question is how  to express, for pure symmetric qubit states of arbitrary number of parties,  the general tensors appearing in the Bloch representation in terms of the Majorana constellation. Furthermore, it would be interesting to explore whether those relations can be extended even beyond symmetric states or even pure states, taking, for example, the generalization of the Majorana representation for arbitrary states found in \cite{Ser.Bra:20}.

\bibliographystyle{ieeetr}
\bibliography{referencesQ.bib}
\end{document}